\documentclass[10pt,journal,compsoc]{IEEEtran}
\ifCLASSOPTIONcompsoc
  \usepackage[nocompress]{cite}
\else
  \usepackage{cite}
\fi
\ifCLASSINFOpdf
   \usepackage[pdftex]{graphicx}
\else
\fi
\usepackage{amsmath}
\usepackage{amssymb}
\usepackage{amsfonts}

\usepackage{nameref}
\usepackage{amsthm}

\newcommand{\kmer}{$k$-mer\xspace}
\newcommand{\kmers}{$k$-mers\xspace}

\newtheorem{definition}{Definition}
\newtheorem{thm}{Theorem}

\newcommand\MyHead[2]{%
\multicolumn{1}{|l|}{\parbox{#1}{\centering #2}}}

\newcommand{\overbar}[1]{\mkern 1.5mu\overline{\mkern-1.5mu#1\mkern-1.5mu}\mkern 1.5mu}

\usepackage{xspace}
\usepackage{algorithm}
\usepackage{algorithmic}
\usepackage{bbm}
\usepackage[flushleft]{threeparttable}
\usepackage{subfigure}
\usepackage{url}

\begin{document}
%
\title{Maximum Likelihood de novo reconstruction of viral populations using paired end sequencing data}
%
%
%
%

\author{Raunaq~Malhotra,
	Steven~Wu, 
        Manjari~Mukhopadhyay,
        Allen~Rodrigo,
        Mary~Poss,
        and~Raj~Acharya,~\IEEEmembership{Fellow,~IEEE}
\IEEEcompsocitemizethanks{
\IEEEcompsocthanksitem R. Malhotra, M. Mukhopadyay, and R. Acharya are with The School of Electrical Engineering and Computer Science, The Pennsylvania State University, University Park, PA, 16802, USA.\protect\\
E-mail: {raunaq@psu.edu}
\IEEEcompsocthanksitem S. Wu is with The Biodesign Institute, Arizona State University, Tempe, AZ, 85281, USA.
\IEEEcompsocthanksitem A. Rodrigo is with The Research School of Biology, Australia National University, Canberra, ACT 2601, Australia. 
\IEEEcompsocthanksitem M. Poss is with The Department of Biology, Pennsylvania State University, University Park, PA, 16802.
}
\thanks{Submitted }}

\IEEEtitleabstractindextext{%
\begin{abstract}
We present MLEHaplo, a maximum likelihood \textit{de novo} assembly algorithm for reconstructing viral haplotypes in a virus population from paired-end next generation sequencing (NGS) data. Using the pairing information of reads in our proposed Viral Path Reconstruction Algorithm (ViPRA), we generate a small subset of paths from a De Bruijn graph of reads that serve as candidate paths for true viral haplotypes. Our proposed method MLEHaplo then generates a maximum likelihood estimate of the viral population using the paths reconstructed by ViPRA. We evaluate and compare MLEHaplo on simulated datasets of 1200 base pairs at different sequence coverage, on HCV strains with sequencing errors, and on a lab mixture of five HIV-1 strains. MLEHaplo reconstructs full length viral haplotypes having a 100\% sequence identity to the true viral haplotypes in most of the small genome simulated viral populations at 250x sequencing coverage.  While reference based methods either under-estimate or over-estimate the viral haplotypes, MLEHaplo limits the over-estimation to 3 times the size of true viral haplotypes, reconstructs the full phylogeny in the HCV to greater than 99\% sequencing identity and captures more sequencing variation for the HIV-1 strains dataset compared to their known consensus sequences. 
\end{abstract}

}

\maketitle

\IEEEdisplaynontitleabstractindextext

%
\IEEEpeerreviewmaketitle

\IEEEraisesectionheading{\section{Introduction}\label{sec:introduction}}

%
%
%
%
\IEEEPARstart{V}{iruses} 
replicating within a host exist as a collection of closely related genetic variants known as viral haplotypes. The diversity in a viral population, or quasispecies, is due to mutations (insertions, deletions or substitutions) or recombination events that occur during virus replication. These haplotypes differ in relative frequencies and together play an important role in the fitness and evolution of the viral population \cite{Boerlijst,bruen2007recombination,domingo,domingo2012viral}.

Viral population diversity was traditionally studied by Sanger sequencing of Polymerase Chain Reaction (PCR) products. This approach has inherent problems that the sample diversity is inadequately assessed, errors are introduced by PCR, and the complete genome of the virus often is not analyzed. The next generation sequencing (NGS) technologies, on the other hand, can sample millions of short reads from the complete viral population. However, they introduce a new challenge of assembling an unknown number of original haplotypes from the short reads. The presence of sequencing errors in short reads in particular confounds the true mutational spectrum in the viral population because viral error rates and sequencing error rates are similar. 
Several algorithms have been proposed for reconstructing haplotypes in a population using a high-quality alignment of short reads to a reference viral genome (See \cite{Niko} for review). Reference based reconstruction of viral haplotypes can be performed either locally or globally. Local haplotype estimation involves inferring haplotypes in short segments along a viral reference sequence to which the reads are aligned \cite{shorah,Prabhakaran}. Error correction is generally performed before \cite{e2,Zagordi,Zagordi2,Prosperi} or along with local haplotype estimation \cite{shorah,QuasiRecomb,Prabhakaran,prosperi2012qure}. Global estimation involves estimating viral haplotypes over segments larger than the length of reads. For global estimation, the reads aligned to a reference or to a consensus sequence are generally arranged in a read-graph, where vertices in the graph are reads and edges denote overlaps amongst the aligned reads. As the number of haplotypes in a population is unknown, an optimal set of viral haplotypes that explains the read-graph is obtained. There are a number of optimization frameworks for analyzing the read-graph \cite{vispa}, for example by modeling it as a network flow problem \cite{westbrooks,viralMCF}, minimal cover formulation, and as a maximum bandwidth paths formulation \cite{Eriksson}. Probabilistic models have also been explored for local haplotype estimation \cite{Zagordi,Prabhakaran}, global haplotype estimation \cite{Zagordi2,Eriksson}, and for analyzing recombinations amongst viral haplotypes \cite{QuasiRecomb}. 

The reference based methods rely on alignment of reads to a reference sequence for error correction, orientation of reads, and for reconstruction of the viral population. A reference sequence can be helpful when it is representative of all haplotypes present in the viral population. However, due to the presence of insertions and deletions, recombination and high mutation rates in some viral populations (e.g. RNA viruses), a large percentage of the reads are unaligned to the reference genome, and are ignored while estimating viral diversity. Additionally, the reference based algorithms are known to have high false-positive rates in reconstruction of viral haplotypes and over-estimate the number of viral haplotypes in datasets with low sequencing diversity \cite{benchmarking,prosperi2013empirical}, which is expected in a virus population obtained from an infected individual. 

\textit{De novo} methods for assembling viral haplotypes provide an alternative to reference-based haplotype estimation, where the viral haplotypes are reconstructed directly from the sampled reads. There are two broad types of \textit{de novo} assembly methods: overlap-layout-consensus and De Bruijn graph based methods. These methods have been widely used for assembly of reads from a single genome. For viral populations, overlap-layout-consensus based \textit{de novo} assembly methods have been used to generate a consensus viral sequence \cite{yang2012novo}. Here, similar short reads are aligned to each other to generate a multiple sequence alignment (MSA) of the reads. A consensus viral sequence is then obtained by a base-by-base majority voting from the MSAs. However, this method assembles consensus sequences of highly diverse viral populations and is not suitable if the aim is to investigate the diversity of the viral population. De Bruijn graph based \textit{de novo} methods first break the short reads into \kmers which form vertices in the graph and edges are drawn in the graph between overlapping \kmers observed in the short reads. For assembly of reads into a single genome, a single path in the De Bruijn graph that visits every vertex can be used to represent the genome sequence. However, reconstructing the viral haplotypes in a population from the De Bruijn graph would involve simultaneously assembling multiple paths that collectively visit every vertex (a graph \textit{cover}) and each path represents a single viral haplotype in the population. The problem of simultaneous assembly from short reads in a single individual is computationally intensive and challenging even for estimating the two paths corresponding to diploid strains in the individual \cite{hapcut}.
 
With the availability of paired-end sequencing data, methods that explicitly incorporate paired end read information in an aligned read graph for reconstructing viral haplotypes have been explored \cite{HaploClique,vga,mangul2014accurate}. Here, instead of computing a minimal set of haplotypes that explains all the paired-reads, the method iteratively merges all the fully connected clusters of reads (max-cliques) in the read-graph to generate viral haplotypes. The method, however, has exponential time complexity in the read coverage \cite{HaploClique}. On the other hand, the problem of reconstructing a minimal set of haplotypes under the constraints of paired-end reads from a read-graph or a De Bruijn graph is NP-hard \cite{rizzicomplexity,dedui}. Thus, only heuristic algorithms for estimating viral haplotypes from paired-end sequencing reads are possible \cite{rizzicomplexity}. 

We present a maximum likelihood based \textit{de novo} assembly algorithm for estimating viral haplotypes using paired-end sequencing data. The main advances made in this paper are (i) we utilize the pairing information of the paired reads to compute a score for paths in a De Bruijn graph that is generated from the overlaps in the reads, (ii) we develop a novel polynomial time heuristic algorithm, Viral Path Reconstruction Algorithm (ViPRA), that generates $M-$ paths corresponding to top $M$ scores through every vertex in the De Bruijn graph, and (ii) we utilize an algorithm MLEHaplo that provides a maximum likelihood estimate of the viral population based on the proposed generative model.

We extensively evaluate our algorithm on simulated datasets varying over sequence diversity, genome lengths, coverage, and presence of sequencing errors. As a replicating viral population consists of viral haplotypes that are derived from a recent common ancestor, our simulated datasets are modeled using a Bayesian coalescent simulator \cite{simcoal,bayessc} which generates sequences derived from a realistic evolutionary model of divergence. Our goal is to successfully reduce the number of paths while retaining the true haplotypes and to recover the viral population accurately and with high precision. Our results demonstrate that MLEHaplo is able to retain the phylogenetic signature in all simulated datasets where there is phylogenetic support for all haplotypes in the population, and it generates full length or near full length haplotypes. We compare our results to existing reference and \textit{de novo} assembly based methods, and observe that MLEHaplo does not suffer from a high false positive rate and a bias towards the reference sequence which occurs in the reference based methods. 

\section{Methods}
\label{methods}
\subsection{Definitions}
A viral population $\mathbf{H}$ is a collection of viral haplotypes $\{H_1,H_2,\ldots, H_P\}$, where each haplotype $H_i$ is a string of length $GS_i$ defined over the alphabet $\Sigma = \{A,G,C,T\}$. The sequence similarity between any two haplotypes $H_i, H_j \in \mathbf{H}$ is assumed greater than 90\% as they are all generated via mutations and recombinations on a recent common ancestral sequence.  A paired read $(R_f,R_r)$ is a pair of two strings over the alphabet $\Sigma$ that are sequenced using NGS technologies from an $IS$ length segment of haplotype $H_i \in \mathbf{H}, i \in \{1,2,\ldots,P\}$. The collection of paired reads sampled from the viral population $\mathbf{H}$ is denoted as $\{(R_{1f},R_{1r}), (R_{2f},R_{2r}), \ldots, (R_{nf},R_{nr}) \}$.  The statistics of insert length $IS$ for the NGS reads are known, although the insert length for a particular paired read is unknown. The average number of times a segment of the viral haplotype is sequenced is known as the coverage of the NGS sequencing. 

A substring $u$ of length $k$ from $H_i, R_f, $ or $R_r$ is known as a \kmer and the reverse complement of the \kmer $u$ is denoted as $\overbar{u}$. The number of times a \kmer $u$ is observed in the sampled reads is known as count of the \kmer $u$. The two \kmers spanning a $(k+1)$ length string in a read are known as consecutive \kmers. A De Bruijn graph is a representation of the sampled reads and consists of \kmers as its vertices. Directed edges are drawn in the graph between consecutive \kmers from the first to the second \kmer. The vertices representing consecutive \kmers are known as consecutive vertices, and a sequence of consecutive vertices in the graph is known as a \textit{path}. A path starting at a vertex with no edges into it (\textit{source} vertex) to a vertex with no edges going out of it (\textit{sink} vertex) is known as a \textit{source-sink} path in the graph.

\subsection{Pre-processing of reads: Error Correction, De Bruijn Graph Construction and Paired Constraints Set}
\label{dbgraph}
\subsubsection{Error Correction} 

In this paper, error-free paired reads are simulated first for calibration of the software. For NGS simulated reads and reads from real sequencing studies, error correction software BLESS \cite{bless} was used for error correction. As this software has been developed for reads sequenced from a single genome, we perform additional error-correction when reconstructing paths in the De Bruijn graph. A path that either begins or ends in vertices corresponding to \kmers that have \kmer counts less than the threshold for BLESS correction (also known as tips in the De Bruijn graph) are removed. These vertices typically correspond to sequencing errors, and this technique has parallels in existing softwares for assembly of single genomes \cite{spades}. 

\subsubsection{De Bruijn Graph} The \kmers from the error corrected paired reads are represented in a De Bruijn graph. As the orientation of the sequenced reads are unknown, \kmers from the paired reads and their reverse complements are stored in the graph. In order to reduce storage space, the De Bruijn graph $G$ is converted into a condensed graph $G_c$ in which linear chains of vertices are condensed into a single vertex, while preserving the edge relationships in the graph $G$. This technique has also been used in the assembler SPADES for single genomes \cite{spades}. 

For viral populations, specially for RNA viruses, it is possible to obtain an acyclic De Bruijn graph for reasonable values of $k$ (around 50-60) as the length of repeats in the viral genomes are typically small (except for terminal repeats in some viral genomes). Two properties of a directed acyclic De Bruijn graph $G$ or the condensed graph $G_c$ are particularly useful for reconstructing sequence of viral haplotypes. First, the insert size $IS$ of a particular paired read can be computed using the distance $d(u,v)$ between two vertices $u$ and $v$ corresponding to the beginning and end of the paired read. This distance $d(u,v)$ is defined as the length of strings spelled by the shortest $u-v$ path in the graph, which can be computed uniquely for a directed acyclic graph. Second, a haplotype of the viral population is spelled by a \textit{source-sink} path in the graph $G$ or $G_c$. A collection of such \textit{source-sink} paths that spans all vertices (a path \textit{cover}) represents the observed viral population. 

\subsubsection{Paired constraints set (PS)}  A Paired Constraints Set $PS$ is generated from the paired reads and contains a list all pairs of \kmers observed in the paired reads along with the number of times a \kmer pair is observed. The set $PS$ for graph $G_c$ can be computed once in time linear in the number of reads and quadratic in the length of the reads.
As there are millions of \textit{source-sink} paths in a typical graph for viral populations, the pairing information of the paired reads is used to guide their selection. The elements of the paired set $PS$ indicate whether a pair of \kmers appear together in one of the viral haplotypes in the population. 

\subsection{Scoring \textit{source-sink} paths in the graph using the set $PS$}
Consider a path $P= (u_1,u_2,u_3,\ldots,u_r)$ in the graph $G$, where $u_1$ is the \textit{source} vertex, $u_r$ the \textit{sink} vertex. The vertices ($u_1,u_2,u_3,\ldots,u_r$) correspond to the \kmers in the sampled reads. Assuming that all paired \kmers from the viral population separated by insert size distance $IS$ are sampled in the observed reads, we define a score $S(P)$ for the path $P$ using the elements of the paired set $PS$ as follows: 

\begin{multline}
\label{eq:06}
S(P) = \frac{1}{E(P)}\cdot \\ 
\sum_{(r,s) \in P \cap [d(s)-d(r)<IS]}{ \{ \mathbbm{1}{[(r,s) \in PS}]  - pen \cdot \mathbbm{1}{[(r,s) \notin PS}] \} }
\end{multline}

The score for the path $P$ is defined over vertex-pairs  $\{ (r,s) ; r \in P \mbox{ \& } s \in P \}$ that are within distance $IS$ on the path $P$. It is proportional to the number of vertex-pairs present in $PS$. The score is also penalized by a penalty term $pen$ for every vertex-pair in $P$ within $IS$ that is not present in $PS$, as under assumption of high coverage, all pairs of \kmers from true viral haplotypes within distance $IS$ would have been observed in the paired reads. The score is normalized by $E(P)$, the expected number of vertex-pairs in a path $P$ that are within the insert size $IS$. 

Paths in the graph that have high scores are candidates for possible viral haplotypes (See supplementary text for rationale). Moreover, a path \textit{cover} of the graph consisting of such paths is a candidate for the viral population $\mathbf{H}$. Equation \ref{eq:06} can be modified in the presence of sequencing errors, wherein, instead of membership in $PS$ within a distance $IS$, the contribution of a vertex pair $(u,v) \in PS$ is weighted by its relative frequency of occurrence. 

\subsection{ViPRA: A heuristic algorithm for estimating top $M$ paths per vertex}
\label{dynamicprogramming}
We propose Viral Paths Reconstruction Algorithm (ViPRA), a heuristic algorithm that computes a path \textit{cover} of the graph using high scoring $M$ paths per \textit{source} vertex in the graph $G$. ViPRA computes $M$ high scoring paths through a vertex using precomputed $M$ high scoring paths through the vertex's neighbors. It starts with the \textit{sink} vertices in the graph and iteratively builds on them using memoization of $M$ paths through each vertex to generate high scoring \textit{source-sink} paths in the graph. This is possible as the score for a path $P$ can be constructed using the scores of its sub-paths starting at a particular vertex to the end (See Algorithms 1, 2 in supplementary data).

The goal of the scoring mechanism $S(P)$ is to ensure that the paths corresponding to true haplotypes have a high score and thus they are retained in the top paths for a vertex as the algorithm propagates from the \textit{sink} vertex to the \textit{source} vertex. Thus, the choice of $M$ is important as it would affect the number of paths generated per \textit{source} vertex. For vertices that are not spanned by paths from the \textit{source} vertices, additional $M$ paths are generated through these vertices. 

\subsection{Maximum Likelihood Estimate of the viral population}
\label{generativemodel}
The path \textit{cover} generated by ViPRA is an over-estimation of the possible paths that represent the viral population. We generate a maximum likelihood estimate of the population using these paths as a starting point. 

The generative model for sampling paired reads from the viral population $\mathbf{H}$ is as follows: For a paired-read $(R_f,R_r)$ of insert length $IS$, it can be sampled from only a single location in the viral haplotype under the assumption that there are no long repeats in viral haplotypes. Thus, the probability of observing the paired read $(R_f,R_r)$ from the viral population $\mathbf{H}$ can be expressed as the ratio of number of haplotypes that share such a read segment to the total number of locations from which any paired read of length $IS$ can be sampled:

\begin{equation}
P((R_f,R_r)|\mathbf{H}) = \frac{q}{P\cdot(GS-IS+1)}=z_{R}
\label{eq:01}
\end{equation}

In this equation, $q$ is the number of haplotypes in $\mathbf{H}$ that have $(R_f,R_r)$ as their sub-string, $P$ is the number of haplotypes in $\mathbf{H}$, and $GS$ is the average length of the viral haplotypes. It should be noted that Equation \ref{eq:01} also holds if $(R_f,R_r)$ denotes a \kmer pair. 

Thus, for a collection of paired reads $\{(R_{1f},R_{2r}),(R_{2f},R_{2r}),\ldots, (R_{nf},R_{nr}) \}$ where $(R_{if},R_{ir})$ is sampled $c_i$ times, assuming independent sampling, the joint probability of observing the paired reads given the viral population $\mathbf{H}$ can be expressed as a multinomial expression:

%

\begin{multline}
P(\{c(R_{1f},R_{1r})=c_1,\ldots,c(R_{nf},R_{nr})=c_n\} | \mathbf{H}) =  \\ \frac{M!}{c_1!\cdot c_2! \ldots c_n!} \prod_{i=1}^{n} z_{R_i}^{c_i}
\label{eq:03}
\end{multline}

where, $M = \sum_{i=1}^{n}c_i = n$. 


Given the above formulation, a maximum likelihood set of haplotypes $\mathbf{H_{ml}}$ can be estimated using equation \ref{eq:03} as follows:

\begin{equation}
\mathbf{H_{ml}} = \max_{\forall \mathbf{H}} P(\{c(R_{1f},R_{1r})=c_1,\ldots,c(R_{nf},R_{nr})=c_n\} | \mathbf{H})
\label{eq:05}
\end{equation}

\subsection{Backward Elimination for estimating $\mathbf{H_{ml}}$}
\label{backwardelimination}
The top $M-$paths through all the \textit{source} vertices is used as candidates for possible haplotypes for computing the maximum likelihood set of haplotypes. The maximum likelihood set of haplotypes for a given dataset is computed using backward elimination. The likelihood computation begins with all the top $M-$paths through the \textit{source} vertices and iteratively remove one path from the set of all paths until the likelihood of the remaining paths in the set starts to decrease. The remaining set of haplotypes constitute a maximum likelihood estimate of the viral population. 

\subsection{Simulation Data Generation}
Viral haplotypes in a viral population are related to one another due to mutations or exchange of genome segments that accrue during replication. Thus, the viral haplotypes share recent common ancestor sequences from which all the viral haplotypes have originated. In order to model this, we generate the simulated viral populations using Bayesian Serial SimCoal Simulator (BayeSSC) \cite{simcoal,bayessc}. The simulator takes population parameters such as population size, mutation rates, and genome length as input and generates a set of sequences that have mutations with respect to a recent common ancestor sequence. The error-free Illumina sequencing paired-reads are simulated from these viral haplotypes using dwgsim (\url{https://github.com/nh13/DWGSIM}), which are used as input to evaluate our proposed method. 

\section{Results}
\label{results}
De Bruijn graph based methods are amenable to address the problem of recovering viral haplotype diversity from shotgun or amplicon sequence data if the challenges discussed above can be solved. Thus an overview of our approach is as follows. As sequencing errors increase the complexity in a De Bruijn graph, we first error correct the sequenced paired reads using existing error correction softwares and store the \kmers from the error corrected reads in the De Bruijn graph $G$, where the vertices are \kmers and consecutive \kmers in a read form edges (Figure~\ref{fig:01}). The vertices in the graph with no incoming edges are known as \textit{source} vertices, while vertices with no outgoing edges are known as \textit{sink} vertices. The paths in the graph starting at a \textit{source} vertex and ending in a \textit{sink} vertex are candidates for viral haplotypes. In order to prune through the millions of \textit{source-sink} paths in the graph, we store in a paired set the pairing information of \kmers observed in the error corrected reads along with their numbers of occurrence. Our proposed Viral Path Reconstruction Algorithm (ViPRA), a heuristic polynomial time algorithm, uses a parameter $M$ and the evidence from the paired set to retain a small number of \textit{source-sink} paths as candidate haplotypes. We limit the over-estimation of number of haplotypes in the viral population by our proposed maximum likelihood estimator, MLEHaplo. MLEHaplo estimates the likelihood that the reads obtained are derived from a set of full-length sequences in the viral population using the paths reconstructed by ViPRA as putative candidates for these sequences.

\begin{figure}[!t]
\centering
\includegraphics[width=3in]{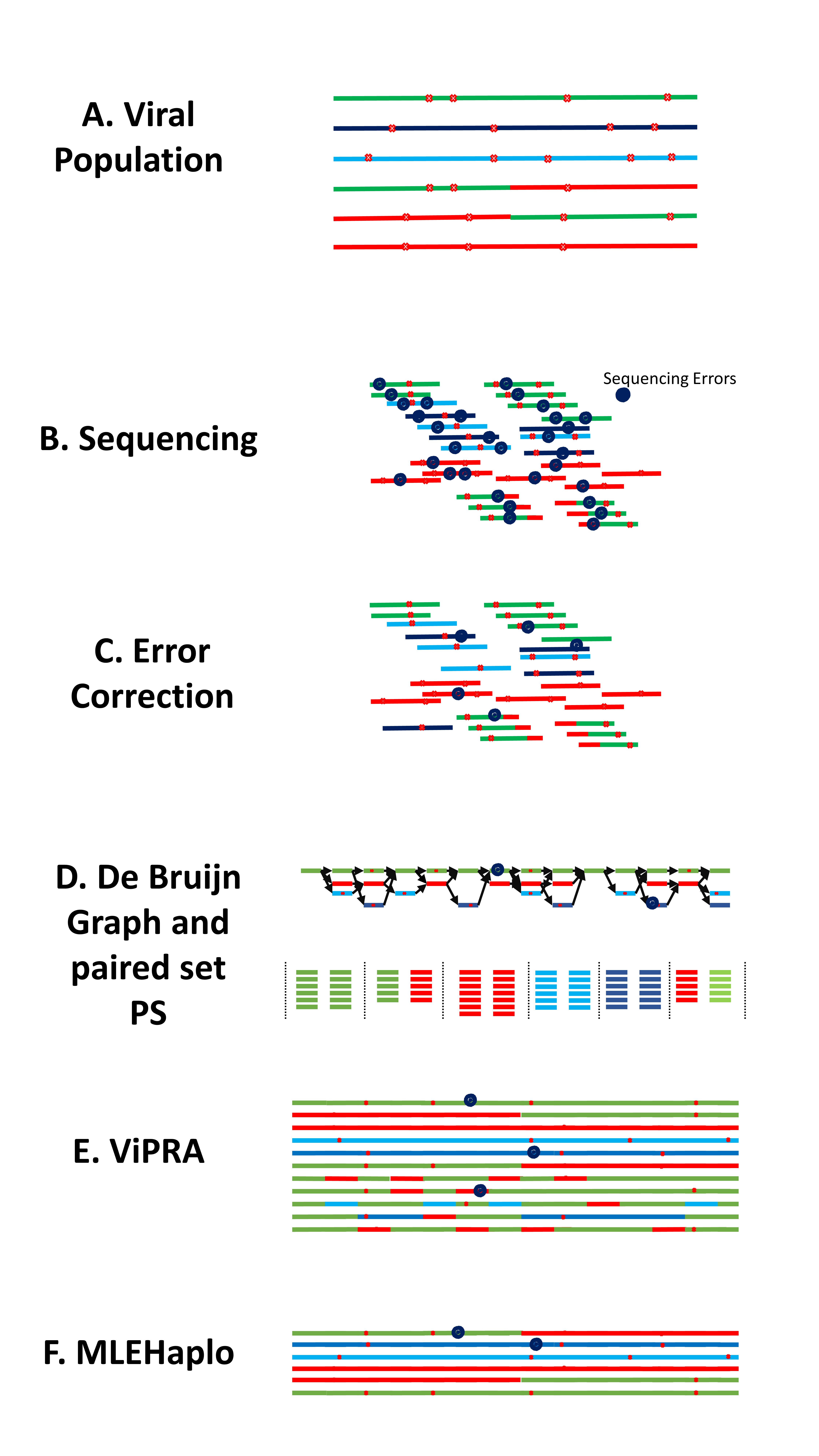}
\caption{ \textbf{Reconstruction of viral haplotypes using PE data:} (A) A viral population consisting of four viral haplotypes (red, green, dark blue, blue lines) which mutations (red dots) and a recombinant haplotype (red-green line) are depicted here. (B) Sequencing technology generates paired-reads with sequencing errors (black dots on short segments from reads). (C) Error correction using existing softwares generates error-free reads (D) The error-corrected reads are stored in a De Bruijn graph. The pairing information of a paired read constitutes all \kmer pairs observed in it, and it is stored in the set $PS$. Colors denote \kmers obtained from the corresponding viral haplotypes in (A). (E) The proposed heuristic algorithm ViPRA reconstructs only a fraction of the total number of paths in the De Bruijn graph, where paths reconstructed have a high score that measures the support of \kmer pairs from a path found in the set $PS$. (F) MLEHaplo reconstructs a maximum likelihood estimate of the viral population, where the sampled reads are modeled to be sampled from a viral population consisting of paths generated by ViPRA as the starting point. \label{fig:01} }
\end{figure}

\subsection{Simulation studies for partial genome reconstruction of viral populations}
Our first evaluation is on simulated datasets to test ViPRA's ability to reduce the total number of paths for possible viral haplotypes and to retain the paths for true haplotypes. The goal of MLEHaplo is to accurately reconstruct the true haplotypes from the output of ViPRA, while limiting the over-estimation of the viral population. The length of the simulated haplotypes are 1200 base pairs (bps), which resembles the length of a gene or a segment in commonly known viruses used for evolutionary studies. Error-free paired end reads are simulated from viral populations that were generated from a common ancestral sequence using a coalescent approach \cite{bayessc,simcoal} at a mutation rate of $10^{-5}$ per generation per nucleotide and a population size of 5000, which are typical evolutionary parameters for a replicating virus population \cite{drake,drake1993rates,sanjuan2010viral}. The average sequencing depth of the 150 bp paired reads is 250x and average insert size is 230 bp (standard deviation of 75 bp). Ten viral populations (denoted as D1-D10) were generated using a different seed sequence under the same evolutionary parameters, each of them containing seven phylogenetically related haplotypes. 

In order to reconstruct the viral haplotypes, the paired reads from each dataset (D1-D10) are stored in De Bruijn graphs with \kmer size of 60, while pairs of \kmers are stored in the paired set. Although the ten datasets are generated under the same evolutionary parameters, there is a 1000-fold difference in the total number of \textit{source-sink} paths present in the error-free De Bruijn graphs for these datasets (Table~\ref{Tab:totalpaths}). This indicates that the complexity of a graph for a viral population varies significantly even when there are no sequencing errors and under the same evolutionary parameters. 

\subsubsection{Evaluation of ViPRA with varying $M$} 
Our algorithm ViPRA prunes through millions of paths in the graph to generate a set of paths having a high paired-end support score, as explained below, that form candidates for viral haplotypes. It takes as input a parameter $M$, the De Bruijn graph $G$, and the paired set $PS$ for a given sequencing data. The output of ViPRA is a path \textit{cover} of the graph containing $M$ \textit{source-sink} paths for each \textit{source} vertex in the graph $G$. The $M$ paths per \textit{source} vertex are ranked based on their support present in the paired set $PS$. A negative penalty is applied to a path when there is no support found for its \kmer pairs within insert size $IS$ in the set $PS$ and we only consider paths with positive scores for evaluation. Additional $M$ paths are also generated through vertices not visited by any of the top $M$ paths generated from the \textit{source} vertices to obtain the path \textit{cover} of the graph. 

We assess the number of paths reconstructed by ViPRA with the variation in parameter $M$. The number of paths recovered is independent of the total number of \textit{source-sink} paths in $G$ for a given dataset but is directly proportional to the parameter $M$ and overall complexity of the graph $G$ (Figure~\ref{fig:02}, Table~\ref{Tab:totalpaths}). Paths for full length and partial length haplotypes are obtained for $M=(5,10)$, while full length paths are obtained for $M\geq15$ (Figure~\ref{fig:02}) indicating that not all vertices were visited by the \textit{source-sink} paths for small values of $M$. We also evaluate the ratio of the true haplotypes recovered by ViPRA to the number of true haplotypes present in the viral population, or \textit{recall}. There is $100\%$ recall of the seven true haplotypes in all the ten datasets for $M \geq 10$, and in nine out of ten datasets for $M=5$. This suggests that $M=10$ is sufficient for recovering the true haplotypes for datasets D1-D10. The paths reconstructed at $M=10$ are used by MLEHaplo for the maximum likelihood estimation.

\begin{figure*}[!t]
\centering
\includegraphics[width=5in]{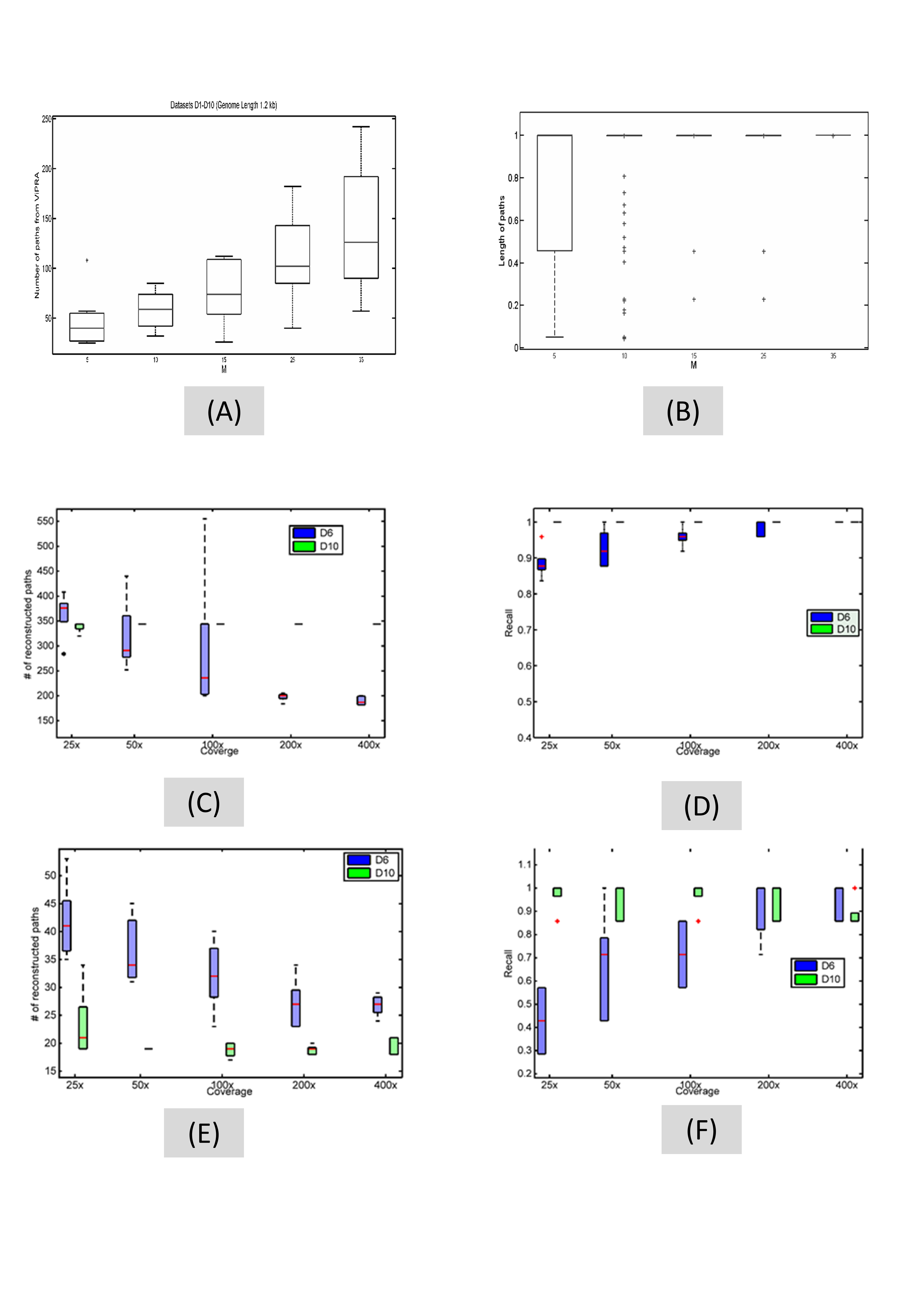}
\caption{\textbf{Evaluation of ViPRA and MLEHaplo on error-free data.} (A) A boxplot of the number of paths obtained from ViPRA when $M$ ranges between $5-35$ for datasets D1-D10 (B) A boxplot of length distribution of paths generated by ViPRA, presented as a fraction of the total length of 1200bp, for datasets D1-D10. Almost all the paths are full length for $M \geq 15$. (C-F) Performance of ViPRA and MLEHaplo with variation in sequencing coverage, at 25x, 50x, 100x, 200x, and 400x denoted on the x-axis of each figure, for datasets D6 and D10. (C) Boxplot of the total number of paths recovered by ViPRA in datasets D6 (blue) and D10 (green). At 50x or greater coverages, ViPRA reconstructs 344 paths for D10 (shown as a line at 344). (D) Boxplot of recall of true haplotypes for datasets D6 (blue) and D10 (second column at each coverage) in paths recovered by ViPRA. The recall of ViPRA is 100\% at all sequencing coverages for D10, thus appearing as a line at 1. (E) Boxplot of number of paths in the set $\mathbf{H_{ml}}$ reconstructed by MLEHaplo in datasets D6 (blue) and D10 (green). (F) Boxplot of MLEHaplo's recall for true haplotypes in datasets D6 (blue) and D10 (green). MLEHaplo's recall for D10 is $>90\%$ at all coverages, the recall for D6 increases with increasing coverage. Recall is defined as the fraction of the number of true haplotypes that have a reconstructed path with exact sequence match.\label{fig:02}}
\end{figure*}

\subsubsection{Evaluation of MLEHaplo} MLEHaplo generates a maximum likelihood estimate $\mathbf{H_{ml}}$ of the viral population using the paths reconstructed by ViPRA. It reduces the number of paths from ViPRA by iteratively removing one haplotype by backward elimination. The shape of the likelihood surface indicates that the likelihood is maximized when the set $\mathbf{H_{ml}}$ contains $7-10$ paths which includes the true number of viral haplotypes in datasets D1-D10 (Figure S1). MLEHaplo has 100\% recall in 6 out of 10 datasets, and recovers six out of seven true viral haplotypes in the remaining datasets with an exact sequence match to the true haplotypes (Table~\ref{Tab:totalpaths}).

\subsubsection{Phylogenetic content of the maximum likelihood estimate}
We compare the phylogeny of the true haplotypes to the set of reconstructed haplotypes $\mathbf{H_{ml}}$ for those datasets with less than 100\% recall. A bootstrap neighbor joining tree is generated using the predicted set $\mathbf{H_{ml}}$ and true set of viral haplotypes for datasets D1, D2, D6, and D8 (Figure S2). In these four datasets, six of the seven true haplotypes have a predicted haplotype with exact sequence match. For the seventh true haplotype, the set $\mathbf{H_{ml}}$ contains predicted haplotypes (hollow green circles) that cluster with the true haplotype (hollow red circles) in the phylogenetic tree with near zero branch lengths. The data also demonstrates that MLEHaplo recovers haplotypes that recapitulate the phylogeny of the true viral population even when haplotypes are closely related to one another, as evidenced by the shorter branch lengths in D7 and D10 (Figure S3). Thus, even when MLEHaplo overestimates the number of actual haplotypes in a population, the phylogenetic information for the true set is retained.

\subsubsection{Effect of variation in coverage}
We assess the performance of ViPRA and MLEHaplo at varying sequencing coverages, as the recovery of low frequency variants may be affected in the presence of dominant viral strains in a population. Error-free Illumina sequencing paired reads of length 150 bp each with an average insert size of 430 bp (standard deviation 75 bp) are simulated from datasets D6 and D10 at coverages of 25x, 50x, 100x, 200x, and 400x. The insert size chosen here is close to the average insert size for newer Illumina sequencing technologies. We compare datasets D6 and D10 because they represent viral populations containing haplotypes with short leaf branch lengths and long leaf branch lengths and have a 1000-fold difference in the total number of \textit{source-sink} paths (Table~\ref{Tab:totalpaths}, Figure S2, Figure (A), (B) in S6). Thus viral haplotypes in dataset D6 are more closely related than those in dataset D10. At each sequencing depth, we simulated five replicates of paired end Illumina sequencing from D6 and D10 and report the average performance over these replicates. 

As the performance of ViPRA at low sequencing coverage is unknown, the parameter $M$ is relaxed from its sufficient value of 10 demonstrated above for D1-D10 and set to $M=35$. For the more diverse viral population D10, ViPRA reconstructs 344 paths and this number is independent of sequencing coverage. The recall of true haplotypes is 100\% in D10 at all coverages, indicating that even 25x coverage is sufficient for ViPRA to reconstruct all the true haplotypes in diverse viral populations (Fig. 2). In low complexity dataset D6, the number of reconstructed paths decreases as coverage increases, and converges to around 200 at coverages greater than 100x (Fig. 2). The increase in number of paths reconstructed by ViPRA at coverages less than 100x is due to reconstruction of both partial length and full length paths in D6 (Figure S4). For low diversity dataset D6, the median recall is close to 90\% at 25x coverage and increases to a median value of 100\% at 200x and 400x coverage. 

MLEHaplo generates the set $\mathbf{H_{ml}}$ containing a median value of around 20 haplotypes for dataset D10 at all sequencing coverages, while the number for D6 decreases from 41 at 25x coverage to 26 at 400x sequencing coverage (Fig. 2), indicating that performance in low diversity populations (similar to dataset D6) improves with increasing coverage. The decrease in number of paths at convergence in D6 can be attributed to the increase in support for \textit{source-sink} paths in the set $PS$. The set $\mathbf{H_{ml}}$ does over-estimate the number of true viral haplotypes by $3-6$ times the size of the viral population, although the over-estimation decreases for D6 with increasing coverage. The median value of recall of the true haplotypes increases from 45\% at 25x coverage to 100\% at coverages greater than 200x in dataset D6. For dataset D10, MLEHaplo has more than 90\% recall at each coverage (Fig. 2). Thus, 200x coverage is sufficient to achieve a median value of recall of 100\% in both datasets. 

\subsection{Comparison to existing methods}
We compare the haplotypes predicted by MLEHaplo on datasets D1-D10 (Results in Table~\ref{Tab:02}) to those obtained from the softwares ShoRAH \cite{shorah}, QuasiRecomb \cite{QuasiRecomb}, and PredictHaplo \cite{Prabhakaran}. The consensus sequence obtained by the \textit{de novo} assembler Vicuna  \cite{yang2012novo} was used as the reference sequence for each method as they require a reference sequence. Through personal communication, the corresponding authors of the softwares VGA \cite{vga} and HaploClique \cite{HaploClique} indicated that their softwares were not being actively developed, thus were not included in comparison. ShoRAH and QuasiRecomb over-estimate the number of predicted haplotypes, while PredictHaplo under-estimates the number of viral haplotypes in eight out of ten datasets (Table~\ref{Tab:02}). QuasiRecomb predicts greater than 1900 haplotypes for all datasets that only contain seven true haplotypes, although a large number of haplotypes are reported as low frequency variants that have relative population frequencies less than $5\cdot10^{-4}$. None of these methods have 100\% recall for any of the ten data sets and all fail to recover any true haplotypes in at least one of the datasets. In contrast, MLEHaplo predicts 6-10 haplotypes for all the ten datasets and correctly reconstructs $6-7$ out of seven true haplotypes in all the ten datasets. Thus, MLEHaplo retains the correct sequences of the haplotypes while accurately predicting the number of haplotypes in each viral population.

\subsection{Performance on HCV populations under the presence of sequencing errors}
Our simulated data were generated under a realistic evolutionary model but fail to capture the constraints on viral sequences that results in a non-random distribution of substitutions in the viral population. Keeping this in mind, we assess the performance of ViPRA and MLEHaplo on HCV populations. The E1/E2 genes (length 1672 bp) from ten HCV strains observed in patient C from a previous study are used as our viral population \cite{hussein2014new}, where the distribution of conserved and variable sites in the E1/E2 genes is not uniform across the length of the genome (Figure S5). The strains are named 1C-10C, out of which strains 2C and 9C are distantly related to the other eight strains (Figure S6 (C)). Strains 3C, 4C, and 10C are either identical or differ from each other by 1 bp and have near-zero branch length when constructing a neighbor joining phylogenetic tree of the ten strains (Figures 3(A) and 3(B)). Thus, all methods are assessed for the recovery of eight E1/E2 genes in the HCV strains. 

A dataset HCV-U is generated by simulating paired reads from all the ten strains at equal frequencies using the software SimSeq \cite{earl2011assemblathon}. Another dataset is generated from the same HCV strains where the relative abundances follows power law with factor 2, denoted as dataset HCV-P. The average sequencing coverage is around 500x with insert sizes 300 bp and read lengths 150 bps. The reads in the datasets HCV-U and HCV-P also include sequencing errors to determine how it affects the performance of ViPRA and MLEHaplo. The software BLESS \cite{bless} was used for error correction and the error corrected reads are represented in the De Bruijn graph. In addition to the error correction from BLESS, the \textit{source-sink} paths in the graph that terminate in vertices with \kmer counts less than the threshold in BLESS (tips in the De Bruijn graph \cite{spades}) are ignored when paths are reconstructed by ViPRA. 

As sequencing errors inflate the number of haplotypes generated by each method, we report the number of distinct paths generated by a method that are different from each other by more than 10 bp, or greater than 99\% sequence difference. All methods except PredictHaplo overestimate the number of haplotypes in the viral population (numbers in brackets Table~\ref{Tab:03}). MLEHaplo and ShoRAH predict similar number of distinct paths, while QuasiRecomb overestimates the distinct paths in the HCV-U dataset. The paths generated by MLEHaplo cluster well to generate a small number of distinct paths for both the HCV-U and HCV-P datasets (Figure S7). The distribution of all pairwise distances for the predicted haplotypes indicates that while MLEHaplo, PredictHaplo, and ShoRAH have similar distributions to the true HCV strains, all of the paths generated by QuasiRecomb have small pairwise distances and the distribution of pairwise distances differs from that of the true HCV strains (Figure S8), which leads to small number of distinct paths in QuasiRecomb. We also compute the recall of the eight true E1/E2 HCV strains by each method to within 10 bp of sequencing difference. MLEHaplo has the highest recall amongst all methods that does not change on the HCV-U and HCV-P datasets, while the recall of other methods decreases in the HCV-P dataset (Table~\ref{Tab:03}). All methods reduce the length of the reconstructed paths, as the terminal ends of the strains have low sequencing coverage and are not reconstructed accurately. 

We assess the quality of the reconstructed haplotypes from each method by constructing a phylogeny of the reconstructed and true haplotypes (Figure~\ref{fig:03}). As mentioned above, the reconstructions of strains 3C, 4C, and 10C are considered together as they differ by 1 bp or less. In the HCV-U dataset, MLEHaplo reconstructs a haplotype within 10 bp sequence difference for six of the seven distinct strains and also one haplotype for the strains 3C, 4C, and 10C. Other methods fail to reconstruct haplotypes for one or the other strain in both HCV-U and HCV-P datasets, with PredictHaplo and QuasiRecomb only recovering three out of the eight strains in the HCV-P dataset. The haplotype 6C is recovered by all methods in both HCV-U and HCV-P datasets. All strains except 8C are reconstructed by MLEHaplo in the HCV-U and HCV-P datasets. For 8C strain, MLEHaplo predicts a haplotype that differs from the true strain by 15 bp in both the HCV-U and HCV-P datasets. The performance of MLEHaplo remains relatively unchanged under the power law distribution of strains in the HCV-P datasets, whereas other methods recover fewer haplotypes. Thus, MLEHaplo recovers the full phylogeny of the viral population as before under the uniform and power law distributions of viral strains. 

\begin{figure*}
\centering
\includegraphics[width=5in]{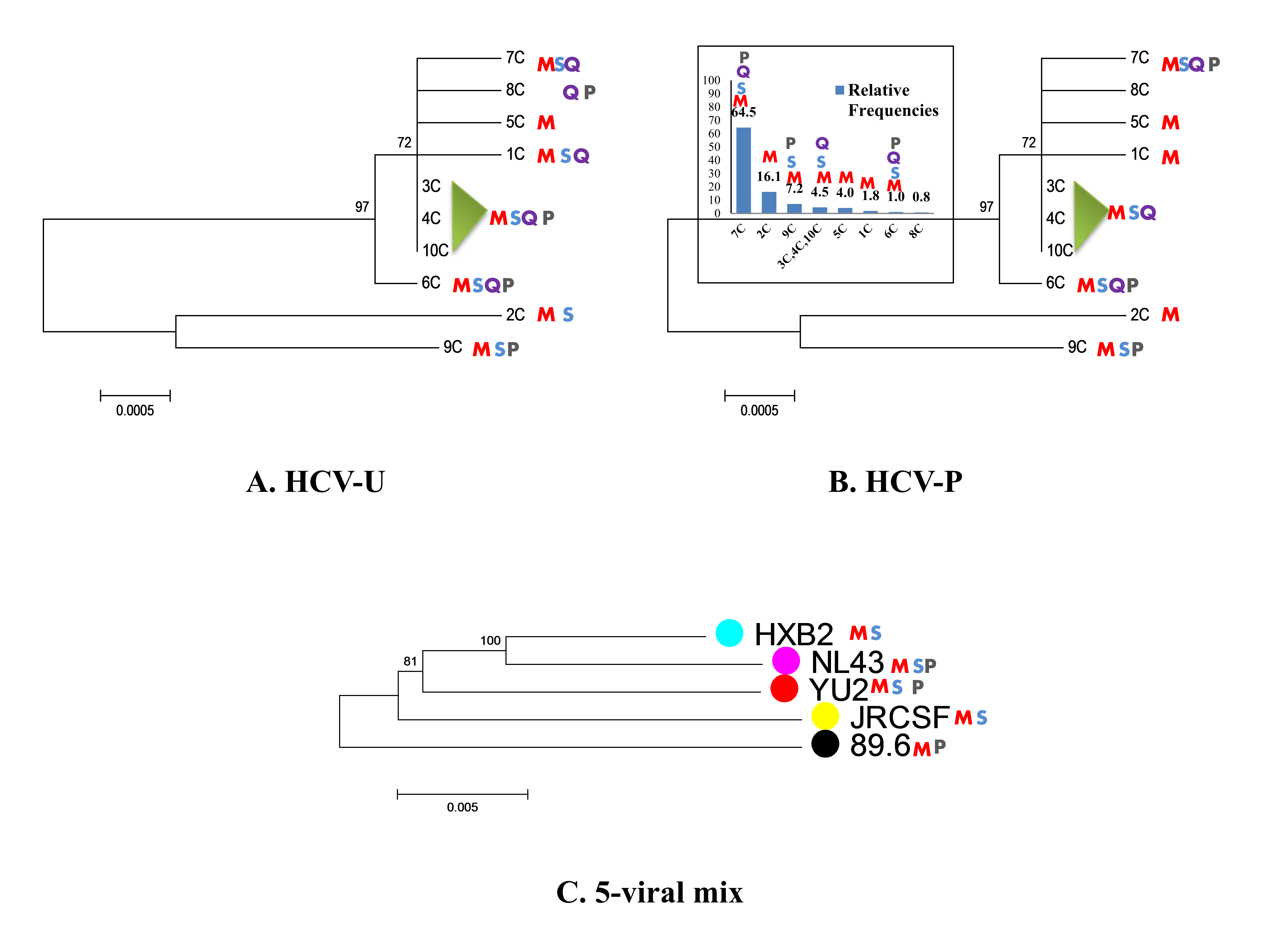}
\caption{ \textbf{Phylogenentic relationships of the reconstructed haplotypes on simulated reads from HCV E1/E2 genes and env gene from a set of 5 HIV-1 strains.} Bootstrap neighbor joining trees are shown for the true haplotypes in HCV population and the set of env genes from 5 HIV-1 strains. The reconstructed haplotypes from MLEHaplo (red M), ShoRAH (blue S), QuasiRecomb (purple Q) and PredictHaplo (dark grey P) that have minimum sequence difference to the true haplotypes in a viral population are indicated next to it. On datasets (A) HCV-U and (B) HCV-P, reconstructed haplotypes that are within 10 bp of the true sequence are reported. Inset shows the relative distributions of haplotypes in the power distributions and the reconstrutions of different methods are replicated on top of the true strains. Reconstruction for strains 3C, 4C, and 10C is considered together as they differ by one bp or are identical. (C) For the five HIV-1 strains, the neighbor joining phylogeny is shown along with the env-gene of an outgroup HIV-1 subtype A strain. MLEHaplo reconstructs a haplotype with greater than 98\% sequence identity to the known consensus sequences for all five strains, while reconstructed haplotypes by ShoRAH and PredictHaplo are also reported. QuasiRecomb ran out of memory on five viral mix dataset and was not evaluated. \label{fig:03}}
\end{figure*}

\subsection{Evaluation on real sequencing data}
We next evaluate the performance of ViPRA and MLEHaplo on a real sequencing data, where there are additional sources of errors during sample preparation and reverse transcription of viruses. A laboratory mixture of five known HIV-1 strains was sequenced using three different sequencing technologies in a previous study \cite{fiveviralmix}, where the consensus sequence of each strain was also obtained (see details in \cite{fiveviralmix}). The strains were named as YU2, NL43, HXB2, JRCSF, and 89.6. In their study, the HIV-1 strains were broken into five overlapping segments of average length 2500 bps that were separately shotgun sequenced using Illumina paired end sequencing. We used Illumina paired end sequencing from the last segment corresponding to the env gene for evaluating our method. 

As the Illumina sequencing data for the five HIV-1 strains was obtained by transfection of 293T cells followed by reverse transcription of the RNA viruses \cite{fiveviralmix}, we expect to see a large number of sequencing differences between the consensus sequences of the viral strains and the reconstructed viruses due to the inherent replication errors of the reverse transcription step. ViPRA generates 262 paths of average lengths 1500 bp from the error corrected env region reads and MLEHaplo reduces it to 78 paths. These paths cover the 2300 bp of the envelope gene in two overlapping 1500 bp segments that can be joined together to reconstruct the full sequence of the envelope gene, as is seen in the tablet software alignments of the paths to the five HIV-1 strains \cite{tablet} (Figure S9). 

We measure whether the reconstructed haplotypes are phylogenetically related to the five HIV-1 strains. We use a known HIV-1 strain subtype A (Accession number AB253635) as an outgroup to generate a rooted neighbor joining phylogeny of the five HIV-1 strains and report the sequence the reconstructed haplotypes that are greater than 97\% of sequencing difference to the consensus sequences of HIV-1 strains (Figure~\ref{fig:03}). As the sequenced reads are generated from replicating viruses, none of the methods reconstruct haplotypes with 100\% sequence identity to env genes. MLEHaplo reconstructs paths for the strains HXB2 and 89.6 that have greater than 99\% sequence identity to the consensus sequences. The strains NL43, YU2, and JRCSF are reconstructed with 97\% sequence identity to their  respective consensus sequences. While ShoRAH reconstructs haplotypes for four strains with 97\% sequence identity, PredictHaplo reconstructs haplotypes for three strains. The method QuasiRecomb ran out of memory on multiple trials and was not evaluated. It should be noted that the reference sequence used for the reference-based methods was generated using the assembler for diverse viral populations Vicuna \cite{yang2012novo}. Thus, while other methods fail to reconstruct all the strains, MLEHaplo generates haplotypes for all the HIV-1 strains. 

We also evaluate whether the reconstructed haplotypes are representative of the sequenced reads by mapping all the \kmers from the reads to the reconstructed haplotypes. The reconstructed haplotypes by ViPRA account for more than 98.46\% of the error corrected \kmers and the paths generated by MLEHaplo represent 97.48\% of the envelope region \kmers. On the other hand, only 95.56\% of the envelope region \kmers are explained by aligning them to the known consensus sequences of the five HIV-1 strains, indicating that the haplotypes reconstructed by MLEHaplo are able to capture more variation observed in the reads than alignment to consensus sequences while capturing the true phylogeny of the five HIV-1 strains. 

\section{Discussion}
\label{conclusions}
We have proposed MLEHaplo, a maximum likelihood \textit{de novo} assembly algorithm for viral haplotypes using paired-end NGS data. The paired reads are represented in a De Bruijn graph and the pairing information of the reads is stored as pairs of \kmers in the set $PS$. The support found in the set $PS$ for pairs of vertices is used to score \textit{source-sink} paths, and the scoring mechanism ensures that the paths represent possible viral haplotypes. As reconstructing a minimal cover of the graph under paired constraints is NP-hard, we have proposed a polynomial time heuristic algorithm, ViPRA, that recovers a small fraction of the total number of paths in the graph. MLEHaplo then reconstructs a maximum likelihood estimate of the viral population using the paths recovered by ViPRA based on a generative model for sampling paired reads from the viral population. MLEHaplo takes about a day and half of running time on a single core machine to process half a million reads, and was comparable in run-time to other methods. 

We have tested ViPRA and MLEHaplo on simulated viral populations that consist of haplotypes arising from common ancestors based on a coalescent model. These simulations are more realistic models of viral populations rather than introducing random mutations uniformly across a known virus. MLEHaplo predicts the smallest set of viral haplotypes and has the highest recall for the true haplotypes when compared to three existing reference based reconstruction methods. 

We also evaluated ViPRA and MLEHaplo on reads simulated from HCV E1/E2 genes identified in an infected patient and on a dataset of five HIV-1 strains which has been used to test haplotype reconstruction methods. MLEHaplo reconstructs haplotypes with greater than 99\% sequence identity to the true haplotypes for HCV E1/E2 genes and greater than 97\% sequence identity for the five HIV-1 strains. The decrease in sequence identity in the case of real HIV-1 strains is understandable as the reads are generated from replicating viruses where one would observe additional sequence variation with respect to the reference. This has been documented as single nucleotide polymorphisms with respect to the reference in the original study \cite{fiveviralmix} and as an increase in sequence alignment score in a previous method \cite{viquas}. Nevertheless, the haplotypes reconstructed by MLEHaplo better explained the observed error corrected sequenced data than the consensus sequences of the known strains, suggesting they form a better representation of the sequenced data. Overall, MLEHaplo is able to generate haplotypes that capture the true phylogeny of the sequenced data. 

The usage of De Bruijn graphs for viral haplotype reconstruction has a number of advantages: The \kmers as vertices avoids the costly computations of reads overlaps. Also, because the De Bruijn graph construction is \textit{de novo}, it contains all the variation observed in the viral population which can be lost by aligning reads to a reference genome. Choosing $k$ greater than $D$, the size of the largest repeat in the viral genome, ensures that the De Bruijn graph obtained is acyclic. As repeats in RNA viruses are generally short, acyclic De Bruijn graphs can be readily obtained for viral populations. However, the De Bruijn graph is sensitive to the presence of errors, and an efficient error correction algorithm is essential. Error correction using the error correction software BLESS \cite{bless} in addition to removal of \textit{source-sink} paths that had low coverage at their ends aided in the reconstruction of full length paths from the algorithm ViPRA. The concept of removing low coverage paths from the graph has been used earlier in whole genome assembler SPADES \cite{spades} where tips and bubbles are collapsed using similar techniques. 

The set $PS$ stores the \kmer pairs observed in paired reads and it can also combine pairing information from multiple overlapping paired reads that are within the insert size $IS$. Efficient storage and computation are possible using binary representation for \kmers and parsing the reads in multiple iterations \cite{dsk,kmc2}. Additionally, this computation needs to be only performed once for a given sequencing dataset. The usage of pairing information for reconstruction of paths from De Bruijn graphs has been shown to improve performance for single genome assembly \cite{spades,paireddbg}. We have used a simplistic version where the number of occurrences of a pair of \kmers is the only information required for the path scoring. The presence of variable insert size mate-pair data or longer length reads can provide additional support for \kmer pairs in the paired set and further improve the performance of ViPRA. 

The proposed scoring mechanism for paths in the De Bruijn graph ensures that paths corresponding to the true viral haplotypes have a high score. Under the assumption of \textit{sufficient} coverage across a viral haplotype, a path in the De Bruijn graph corresponding to a true viral haplotype should have high paired \kmer support within the insert size $IS$ in set $PS$ and thus will have a high score. Penalizing the score for \kmer pairs that are missing only within the insert size length $IS$ is essential in discarding false positive paths, as the absence of such \kmer pairs has a low probability under \textit{sufficient} coverage. On the other hand, we do not expect to see \kmer pairs that are at distances greater than insert size due to our sequencing process, and thus such \kmer pairs are not penalized. This ensures that when ViPRA retains the $M$ top scoring sub-paths at a vertex these sub-paths correspond to segments of the true haplotypes and as ViPRA propagates to the \textit{source} vertices of the graph, ViPRA reconstructs paths corresponding to the true haplotypes. 

The number of paths generated by ViPRA is dependent on the choice of the parameter $M$. The choice of $M$ is dependent on the inherent diversity of the viral population, which can be inferred by the number of vertices in the De Bruijn graph. The relative ratio of the counts of occurrences of most frequent \kmers to the average sequencing coverage can be used as a lower bound for determining $M$. As the most frequent \kmers are likely to be shared amongst the viral haplotypes in the population, such an $M$ is the minimum that is required for ViPRA to reconstruct all viral haplotypes that share this \kmer's vertex. We experimentally determined that $M=10$ was sufficient for reconstructing all the true viral haplotypes in the simulated datasets and real datasets, even in the presence of sequencing errors. 

MLEHaplo generates a maximum likelihood set of paths using the paths generated by ViPRA as a starting point. As the number of viral haplotypes in the population is unknown, MLEHaplo uses a backward elimination optimization that iteratively removes a path from the set of paths reconstructed by ViPRA such that the likelihood of the sampled paired reads under the remaining set of paths increases. Thus, MLEHaplo further reduces the number of paths reconstructed from the graph thereby reducing the false positive paths at convergence. The advantage of MLEHaplo is that all the sampled paired reads are explained by one of the reconstructed paths at convergence. However, it is sensitive to the paths reconstructed by ViPRA and can only improve on this set of paths. Additionally, as a path removed by backward elimination once is never considered again, it leads to the retention of haplotypes that are close to the true sequences but do not affect the likelihood. Removing one path iteratively makes MLEHaplo a time consuming step, and parallelization of the code can improve the speed of computation. An advantage of MLEHaplo and ViPRA over the existing methods is that they retain the correct phylogeny of the true viral haplotypes, even when the reconstructed paths are not an exact match to the true viral haplotypes. Thus, the viral diversity can be correctly inferred from the reconstructed paths. The software implementation for MLEHaplo is available at \url{https://github.com/raunaq-m/MLEHaplo}.

\section{SUPPLEMENTARY FIGURES}
\subsection*{S1 Figure}
\label{fig:01} 
\textbf{Likelihoods of a set of haplotypes reconstructed by MLEHaplo. }
x-axis is the number of haplotypes present in the predicted set, y-axis shows the maximum likelihood for any combination of predicted haplotypes of the size mentioned on x-axis. The vertical bars denote 95\% confidence intervals. The likelihood drops significantly when \kmer-pairs that are present in the reads are not explained by a current size of haplotypes. Inset shows where the maximum likelihood was achieved in datasets D1 and D8. 

\subsection*{S2 Figure}
\label{fig:phylogenytrees}
\textbf{Phylogenetic tree of true and predicted haplotypes for datasets D1, D2, D6, and D8.} 
Bootstrap neighbor joining tree for the predicted (green) and true haplotypes (red) for datasets D1, D2, D6, and D8. Solid triangles indicate perfect sequence match between a reconstructed path and a true haplotype that clusters with it, while hollow green squares indicate predicted haplotypes that have mismatches to the true haplotypes (hollow red circles). Except for D1, D2, D6, and D8, all datasets have a reconstructed path exactly matching each of the seven true haplotypes. D6 has more than one reconstructed paths for three out of seven true haplotypes, nevertheless, they all group together in the neighbor joining tree.

\subsection*{S3 Figure}
\label{fig:restofsixphylogeny}
\textbf{Phylogenetic tree of true and predicted haplotypes for datasets D3-D5, D7, and D9-D10. } 
Bootstrap neighbor joining trees for the predicted (green) and true haplotypes (red) for datasets D3, D4, D5, D7, D9 and D10. Solid triangles indicate perfect sequence match between the reconstructed path and a true haplotype that clusters with it. All six datasets have perfect reconstruction.

\subsection*{S4 Figure}
\label{fig:Length12kb}
\textbf{ViPRA: Effect of variation in coverage in length of recovered paths for small genome size viral populations.}
Boxplot of length of paths recovered  for datasets D6 (first column) and D10 (second column) at coverages of 25x, 50x, 100x, 200x, and 400x. The length of paths are normalized genome size (1200bp). Partial length paths are recovered at low coverages. Full length paths are recovered at coverages greater than 100x. 

\subsection*{S5 Figure}
\label{fig:distributionofmutations}
\textbf{Distribution of variable and conserved sites for simulated viral population and HCV strains.}  Figure shows the distribution of variable sites (black lines) and conserved sites (gray background) across an alignment of 7 haplotypes in D1 dataset (top part), and across an alignment of ten E1/E2 HCV strains. The distribution for D1 is uniform, while it is highly non-uniform for HCV strains. 

\subsection*{S6 Figure}
\label{fig:distancedistribution}
\textbf{Distance matrix for the simulated viral populations and HCV E1/E2 strains.} Figure shows the sequence distance distribution of the viral haplotypes used in simulated datasets under the coalescent model and the ten haplotypes in the HCV viral population. (A) Dataset D6, (B) Dataset D10, and (C) HCV dataset. 

\subsection*{S7 Figure}
\label{hcvmlehaploheatmaps}
\textbf{Pairwise distance of the predicted haplotypes from MLEHaplo for E1/E2 HCV strains.}  The pairwise distance heatmaps for the predicted haplotypes by MLEHaplo are shown for (A) 43 haplotypes reconstructed in HCV-U dataset, and (B) 81 haplotypes reconstructed in HCV-P dataset. The blocks observed in the heatmaps indicate predicted haplotypes that have similar to each other. The colorbar indicates the scale for the pairwise distances. The haplotypes are clustered into 27 distinct paths in HCV-U and 28 distinct paths in HCV-P dataset. 

\subsection*{S8 Figure}
\label{sequencedistributions}
\textbf{Comparison of pairwise distances for reconstructed haplotypes for E1/E2 HCV strains.} (A) The histograms for the pairwise distances between the 10 true E1/E2 HCV strains is shown. (B) The histogram of the pairwise distances of the reconstructed haplotypes by each method in the HCV-U dataset (left column) and HCV-P dataset (right column) are shown. Haplotypes reconstructed by QuasiRecomb have low pairwise distances compared to the true haplotypes, while MLEHaplo, ShoRAH resemble the distribution for the true haplotypes. 

\subsection*{S9 Figure}
\label{mappingfiveviralmix}
\textbf{Alignment of reconstructed haplotypes from MLEHaplo for five HIV clones dataset.} The reconstructed haplotypes are 1500 bp long, but cover the envelope gene of all the five HIV strains with overlapping segments, which can be easily joined together.

\section{Tables}
\begin{table}[!ht]
\centering
\caption{ \bf{Comparison of total number of paths to paths generated by ViPRA at $M=10$}}
\label{Tab:totalpaths}
\begin{threeparttable}
\begin{tabular}{|c|c|c|c|}\hline
Dataset & \MyHead{2cm}{Total paths in the graph} & \MyHead{1.5cm}{ \# of paths from ViPRA \tnote{a,b}} & \MyHead{1.5cm}{MLEHaplo: \# of paths in $\mathbf{H_{ml}}$ \tnote{b}}  \\ \hline
D1	&	4824	&	50(7)	&	7(6)\\
D2	&	44299	&	38(7)	&	8(6)\\
D3	&	325585	&	42(7)	&	7(7)\\
D4	&	164387	&	52(7)	&	7(7)\\
D5	&	6768	&	73(7)	&	7(7)\\
D6	&	1665626	&	32(7)	&	10(6)\\
D7	&	2423	&	66(7)	&	7(7)\\
D8	&	8712	&	74(7)	&	8(6)\\
D9	&	4895	&	75(7)	&	7(7)\\
D10	&	1357	&	85(7)	&	7(7)\\ \hline
\end{tabular}
\begin{tablenotes}
\item[a] The score of each path $P$ is non-negative.
\item[b] Number in bracket indicates the \# of true haplotypes in the population that are retained by ViPRA or MLEHaplo with an exact sequence match. There are seven true viral haplotypes in each dataset. 
\end{tablenotes}
\end{threeparttable}
\end{table}

\begin{table*}[t]
\centering
\caption{ \bf{Comparison of results on D1-D10 to existing methods}}
\label{Tab:02}
\begin{threeparttable}
\begin{tabular}{|c|c|c|c|c|}\hline
Dataset & MLEHaplo  &  ShoRAH \tnote{a} & QuasiRecomb \tnote{a} & PredictHaplo \tnote{a} \\ \hline
D1 & 	 	7(7)		&	36(3)		&	6837 (2) & 	4(0) 	\\
D2 & 		8(6)		&	56(1)		&	9864  (0)	& 	\textendash		\\
D3 & 	 	7(7)		&	49(1)		&	1964   (5)	&	5(4)		\\
D4 & 	 	7(7)		&	39(0)		&	8582 (0)	&	6(1)		\\
D5 & 	 	7(7)		&	\textendash 	&	9988 (0)	&	3(0)		\\
D6 & 	 	10(6)	&	17(4)		&	4265  (3)	&	7(2)		\\
D7 & 	 	7(7)		&	8(2)			&	7161 (1)	&	7(6)		\\
D8 & 		8(6)		&	29(1)		&	9943 (1)	&	3(0)		\\
D9 & 		7(7)		&	14(1)		&	8386 (2)	&	5(1)		\\
D10 & 	 	7(7)		&	19(1)		&	9239 (0)	&	4(0)		\\
\hline
\end{tabular}
\begin{tablenotes} 
\item[*] The number in the bracket indicates the number of true haplotypes that are present in the predicted set with an exact match. 
\item[a] Consensus sequence generated by VICUNA was used as reference sequence for reference based methods SHoRAH, QuasiRecomb, and PredictHaplo. 

\end{tablenotes}
\end{threeparttable}
\end{table*}

\begin{table*}[t]
\centering
\caption{\bf{Comparison of number of paths generated on HCV populations.}}
\label{Tab:03} 
\begin{threeparttable}
\begin{tabular}{|c|c|c|c|c|}\hline
Dataset	&	MLEHaplo	&	ShoRAH 	&	QuasiRecomb 	&	PredictHaplo \\ \hline
HCV-U	&	27(43)	&	27(280)	&	45(1024)	&	7(7)	\\
HCV-U(Recall) \tnote{a} & 87.5\%& 75\% & 62.5\% & 50\% \\ \hline
HCV-P	&	28(81)	&	18(41)	&	8(493)	&	7(8)	\\  
HCV-P(Recall) \tnote{a} & 87.5\% & 50\% & 37.5\% & 37.5\% \\

\hline
\end{tabular}
\begin{tablenotes} 
\item[*] The number in the bracket indicates the \# of paths predicted by a method, while the value outside denotes the number of clusters of paths that are differ by greater than 99\% sequence difference. 
\item[a] Recall is the fraction of the true E1/E2 HCV strains (out of 8) that are recovered with more than 99\% sequence identity (less than 10 bp difference) by a method.  
\end{tablenotes}
\end{threeparttable}
\end{table*}




\appendices
\section{Rationale that scoring mechanism selects paths that correspond to viral haplotypes in the population}

The \kmer pairs in the set $PS$ is a summarization of the observed paired reads and as the scoring of paths is based on the set $PS$, it is useful in selecting \textit{source-sink} paths that represent the viral population. As defined in the text, the score $S(P)$ of a path $P$ is defined as : 

\begin{multline}
\label{eq:06}
S(P) = \frac{1}{E(P)}\cdot \\ 
\sum_{(r,s) \in P \cap [d(s)-d(r)<IS]}{ \{ \mathbbm{1}{[(r,s) \in PS}]  - pen \cdot \mathbbm{1}{[(r,s) \notin PS}] \} }
\end{multline}

We provide a rationale for the paths corresponding to true viral haplotypes to have a high score $S(P)$. 

\begin{definition}[Sufficient Coverage]
The sampled paired reads are defined to have \textit{sufficient} coverage of the viral population $\mathbf{H}$ if there exists paired reads $(R_f,R_r)$ that sample every haplotype $H \in \mathbf{H}$ and there exists a paired read $(R^{\prime}_{f},R^{\prime}_{r})$ that samples a pair of \kmers $(u_i,u_j) \in H$ with $d(u_i,u_j)<IS$.
\end{definition}

\noindent Given \textit{sufficient} coverage and the definition of the paired set $PS$, if two vertices $u_i$ and $u_j$ are present in a paired read, then a path $P$ that contains both vertices $u_i$ and $u_j$ can be a possible viral haplotype. Using \textit{sufficient} coverage, we can show that paths corresponding to the true viral haplotypes in $\mathbf{H}$ have a path score $S(P) = 1$. Thus in order to extract paths from the graph $G$, it is sufficient to focus on the paths with high scores. 

\begin{thm}
\label{thm:01}
Given \textit{sufficient} coverage of the viral population $\mathbf{H}$, for a path $P$ in the graph $G$ that corresponds to a viral haplotype in $\mathbf{H}$, the score $S(P) = 1$. 
\end{thm}

\begin{proof}
The proof can be broken into two parts: 
\begin{itemize}
\item[1.] For every viral haplotype $H_i \in \mathbf{H}$, there exists a path $P$ in the graph $G$, and 
\item[2.] The score $S(P)$ for such paths is 1. 
\end{itemize}

As \textit{sufficient} coverage implies that all the viral haplotypes $H_i \in \mathbf{H}$ are sampled by the paired reads, for a given viral haplotype $H$, it follows that there exists vertices $u_j$ in the graph $G$ that are sampled from the viral haplotype $H$, which implies that a path $P= (u_{1}, u_{2}, \ldots, u_{m})$ , where $u_i \in H$, exists in the graph. 

The score $S(P)$ for such a path, by definition in equation \ref{eq:06}, is the summation over all pairs of vertices $(u_i,u_j) \in P$ that are within the distance $IS$. Again, by the definition of \textit{sufficient} coverage, all pairs of vertices from a viral haplotype $H\in \mathbf{H}$ within distance $IS$ are sampled by some paired read, which implies that $(u_i,u_j) \in PS$. Thus the score $S(P)$ is :

\begin{multline}
S(P) = \frac{1}{E(P)} \sum_{(r,s) \in P \cap d(s,r) < IS} { \mathbbm{1} [(r,s) \in PS] - pen \cdot \mathbbm{1} [(r,s) \notin PS] } \\
 = \frac{1} {E(P)} \sum_{(r,s) \in P \cap d(s,r) < IS} {\mathbbm{1} [(r,s) \in PS]}  = \frac{1} {E(P)} \cdot E(P) = 1
\end{multline}

\end{proof}

\section{Viral Path Reconstruction Algorithm (ViPRA)}
Algorithm \ref{alg:01} describes the pseudo-code of ViPRA that reconstructs the top $M$ paths through every vertex in the graph. ViPRA starts by initializing the paths from all vertices $u$ to the \textit{sink} vertex $u_{sink}$ to empty sets. The scores for all paths from each vertex are also initialized to empty sets (Lines 1-4). Next it iterates over all vertices $u \in V_c$ in increasing order of their distance to the sink vertex ($d(u,u_{sink}))$ to compute the top $M-$ paths through each vertex $u$ using the function TOP-M-PATHS-FOR-VERTEX$(u, E_c,PS)$ (Lines 5-10). When a graph has multiple \textit{sink} vertices, a universal \textit{sink} vertex is defined that has an edge from each of the \textit{sink} vertex to it. 

\begin{algorithm}[!htb]
\caption{\label{alg:01}ViPRA() : Top $M$ paths per vertex based on set PS and the graph $G$}
\textbf{Input:}  Directed De Bruijn graph $G(V,E)$, Set $PS$, d(.) the distance between vertex pairs in $G$.\\
\textbf{Output:} $Paths(u) =\{ P_{u1}, P_{u2}, \ldots, P_{uM} \}\mbox{ } \forall u \in V$ \\
	$Score(u) = \{ S_{u1}, S_{u2}, \ldots, S_{uM} \} \mbox{ } \forall u \in V$ \\
\begin{algorithmic}[1]
\FOR { each vertex $u \in V$}
\STATE $Paths(u) =[\emptyset ] $ //Initialize the paths for each vertex
\STATE $Score(u) = [\emptyset ]$
\ENDFOR
\STATE $Score(u_{sink}) = \{0 \}$
\STATE $Paths(u_{sink}) = \{ ``\_'' \} $ // Place holder between vertices
\STATE $N = $ SORT-INCREASING$(V, D(.))$ // Sort vertices in increasing order of distance to the \textit{sink} vertex and store it in $N$
\FOR  { $ i=1, \cdots, |N|$}
\STATE $[Paths(N[i], Score(N[i])) ] = $TOP-M-PATHS-FOR-VERTEX$ (N[i], E, PS) $
\ENDFOR
\end{algorithmic}
\end{algorithm}

\begin{algorithm}[!htb]
\caption{\label{alg:02} TOP-M-PATHS-FOR-VERTEX $(u, E,PS)$ : Top $M-$paths for a vertex $u$ in the graph $G$ }
\textbf{Input:} Vertex $u$, Edge set $E$, Set $PS$, Insert size  $IS$\\
\textbf{Output:} $Paths(u), Score(u)$, Set of top  $M-$paths for the vertex $u$, and their respective scores
\begin{algorithmic}[1]
\STATE $TP = [ \emptyset ]$
\STATE $SP = [ \emptyset ]$
\FOR {each $\{v,  (u,v) \in E\}$}
\STATE $TP =$  JOIN $( TP , Paths(v)) $ //Obtain top $M$ -paths of the neighbor
\STATE $SP = $ JOIN $(SP Score(v) )$ 
\ENDFOR
\FOR {each path $T \in TP$}
	\STATE $l = \mbox{length}(T)$
	\STATE $SP(T) = SP(T) \cdot \frac{l\cdot(l-1)}{2} $
	\FOR {each vertex $w \in T$}
	\IF { $(u,w) \in PS$ }
	\STATE $SP(T) = SP(T) +1$ 
	\ELSIF {$D(w) - D(u) > IS^{*}$}
	\STATE $SP(T) = SP(T) + 1$
	\ELSE
	\STATE $SP(T) = SP(T) - l$
	\ENDIF
	\ENDFOR
	\STATE $SP(T) = SP(T) \cdot \frac{2}{(l+1)\cdot l}$
\ENDFOR

\STATE $(TP,SP)$ = SORT-DECREASING$(TP, SP(.))$ // Sort $TP$ paths based on the corresponding scores $SP$ of the paths
\FOR {$i=1 \ldots M$}
\STATE $Paths(u)=$ JOIN $( \{u ``\_''TP[i]\} , Paths(u))$  // Add $u$ to the path 
\STATE $Score(u) =$ JOIN $( \{SP[i]\} , Score(u))$
\ENDFOR
\STATE \textbf{Return} $Paths(u), Score(u)$ 
\end{algorithmic}
\end{algorithm}

Algorithm \ref{alg:02} describes the memoized algorithm that computes the top $M-$ paths from a vertex $u$ to the \textit{sink} vertex $u_{sink}$ ($TOP-M-PATHS-FOR-VERTEX (u, E, PS) $). It first recovers the top $M-$ paths from all the neighbors of $u$ having an incoming edge from $u$ into the array $TP$ and their scores in $SP$ (Lines 1-6). 
As the distance $d(u,u_{sink})$ is greater than the distance of its neighboring vertex $s$, ($d(u,u_{sink}) > d(s,u_{sink})$), the arrays $Paths(s)$ and $Score(s)$ have already been computed (Algorithm \ref{alg:01} in line 9). The score for each of the path in $TP$ is updated when adding the vertex $u$ to the path. The path scores are updated taking into account memberships in the set $PS$ (Lines 7-20). The penalty term ($pen$) is proportional to the length of the path $T$ (Line 16). The paths stored in the array $TP$ are sorted based on their updated scores (Line 21). The first $M-$paths in the array $TP$ and their scores $SP$ are stored as the paths through vertex $u$ to the \textit{sink} vertex $u_{sink}$ (Lines 22-26). 

\subsubsection{Time complexity analysis of ViPRA} 
ViPRA runs Algorithm \ref{alg:02} as its sub-routine and the time complexity of the two algorithms is evaluated together. Lines 1-7 of ViPRA (Algorithm \ref{alg:01}) takes $O(|V_c|+ |V_c|\log{|V_c|})$ time to initialize and then sort the vertices. The $|V_c|$ calls to the sub-routine TOP-M-PATHS-FOR-VERTEX$ (N[i], E_c, PS) $ take $O(M\cdot |E_c|)$ time (Lines 8-10 of ViPRA, Lines 1-6 of Algorithm \ref{alg:02}). Each edge $(u,v) \in E_c$ is encountered at most $M$ times to store the top $M-$paths for the vertex $u$. Lines 7-20 in Algorithm \ref{alg:02} look at each path $T$ in array $TP$ and query for vertex-pairs $(u,w) \in PS$. The number of queries are proportional to the length of the path $T$. Thus, the total number of queries is bounded by $O(|TP| \cdot |T|)$. As the length of path increases linearly to the maximum depth in the graph $G$ ($O(|V_c|$), and the size of $|TP|$ is bounded by $O(M)$, the total time complexity for lines 7-20 is $O(M\cdot |V_c|^{2})$. As lines 21-26 perform a sorting operation for $|TP|$ elements at a time, its time complexity is bounded by $O(M\cdot|E_c|\log{(L\cdot |E_c|)})$. Thus overall running time of ViPRA is $O(M\cdot |V_c|^{2} + M\cdot |E_c| \log{(M\cdot |E_c|)})$. Notice that the running time of ViPRA increases log linearly with the parameter $M$ of the algorithm, and is quadratic in the input graph parameters, namely $G_c(V_c,E_c)$.

\ifCLASSOPTIONcompsoc
  \section*{Acknowledgments}
\else
  \section*{Acknowledgment}
\fi
This work was supported by the National Science Foundation (Grant numbers 1421908, 1533797 to RA, MP); National Evolutionary Synthesis Center (NesCENT) (through a researcher fellowship to RM).

\ifCLASSOPTIONcaptionsoff
  \newpage
\fi

\end{document}